\documentclass[a4paper, 10pt]{article}
\usepackage{palatino}
\usepackage{amsmath}
\usepackage{amsthm}
\usepackage{amssymb}
\usepackage{latexsym}
\usepackage{multirow}

\newtheorem{theorem}{Theorem}[section]
\newtheorem{definition}[theorem]{Definition}

\newtheorem{proposition}[theorem]{Proposition}
\newtheorem{corollary}[theorem]{Corollary}

\newtheorem{lemma}[theorem]{Lemma}
\newtheorem{construction}[theorem]{Construction}

\def\wt{\textrm{wt}}
\newcommand{\AI}{\textrm{AI}}
\newcommand{\spa}{\textrm{span}}

\title{Constructing and Counting Even-Variable Symmetric Boolean Functions with Algebraic Immunity not Less Than $d$}
\author{Yuan Li, Hui Wang and Haibin Kan}
\date{}

\begin{document}

\maketitle

\begin{abstract}
In this paper, we explicitly construct a large class of symmetric
Boolean functions on $2k$ variables with algebraic immunity not less
than $d$, where integer $k$ is given arbitrarily and $d$ is a given
suffix of $k$ in binary representation. If let $d = k$, our
constructed functions achieve the maximum algebraic immunity.
Remarkably, $2^{\lfloor \log_2{k} \rfloor + 2}$ symmetric Boolean
functions on $2k$ variables with maximum algebraic immunity are
constructed, which is much more than the previous constructions.
Based on our construction, a lower bound of symmetric Boolean
functions with algebraic immunity not less than $d$ is derived,
which is $2^{\lfloor \log_2{d} \rfloor + 2(k-d+1)}$. As far as we
know, this is the first lower bound of this kind.
\end{abstract}

\section{Introduction}
Algebraic attack has received a lot of attention in studying
security of the cryptosystems. If a Boolean function used in stream
ciphers has low degree annihilators, it will be easily attacked.
This adds a new cryptographic property for designing Boolean
functions to be used as building blocks in cryptosystems which is
known as algebraic immunity(AI). Since then algebraic immunity, as a
property of Boolean functions, is widely studied.

Constructing Boolean functions with high AI is interesting and
important. A lot of general methods to construct Boolean functions
with maximum algebraic immunity are proposed
\cite{basictheory_Dalai}, \cite{LiNa1}, \cite{construction_Carlet}.
Results in \cite{LiNa1}, \cite{anewupperbound} show that the number
of general Boolean functions achieving maximum algebraic immunity is
large.

Among all Boolean functions, symmetric Boolean function is an
interesting class and their properties are well studied
\cite{Canteaut_sbf}, \cite{LiNa2}, \cite{Qu1}. In \cite{LiNa2},
\cite{Qu1}, the authors proved that there are only two symmetric
Boolean functions on odd number of variables with maximum AI. In
Braeken's thesis \cite{Braeken}, some symmetric Boolean functions on
even variables with maximum AI are constructed. In \cite{Qu2}, more
such functions are constructed, which generalizes results in
\cite{Braeken}. In \cite{Liao_c2mp1}, by using weight support
technique, all $(2^m+1)$-variable symmetric Boolean functions with
submatrimal algebraic immunity $2^{m-1}$ are constructed.

In this paper, we focus on constructing symmetric Boolean functions
with high algebraic immunity on $2k$ variables, where $k$ is given
arbitrarily. For a given $d$, where $d$ is a suffix of $k$ in binary
representation, we construct a large class of Boolean functions with
AI not less than $d$. Particularly, if let $d=k$, our constructed
Boolean functions achieve maximum AI. Comparing with all the
previous constructions of this kind, the number of our constructed
Boolean functions is much larger. Furthermore, a lower bound of
symmetric Boolean functions with algebraic immunity not less than
$d$ is derived.

\section{Preliminaries}

Let $\mathbb{F}_2$ be the finite field with only two elements. To
prevent confusion with the usual sum, the sum over $\mathbb{F}_2$ is
denoted by $\oplus$. The Hamming weight of a vector $\alpha =
(\alpha_1, \ldots, \alpha_n)$ is defined by $\wt(\alpha) =
\sum_{i=1}^{n}{\alpha_i}$.

A Boolean function on $n$ variables may be viewed as a mapping from
$\mathbb{F}^n_2$ into $\mathbb{F}_2$. We denote by $\mathcal{B}_n$
the set of all $n$-variable Boolean functions. The Hamming weight
${\rm wt}(f)$ is the size of the support ${\rm supp}(f) = \{x \in
\mathbb{F}_2^n \mid f(x)=1\}$. The support of $f$ is also called the
on set of $f$, which is denoted by $1_f$. On the contrary, the off
set of $f$ is the set $\{x \in \mathbb{F}_2^n \mid f(x)=0\}$, which
is denoted by $0_f$. Any $f \in \mathcal{B}_n$ can be uniquely
represented as
\begin{equation}
f(x_1, x_2, \ldots, x_n)  =  \bigoplus_{\alpha \in \mathbb{F}_2^n}
c_\alpha \prod_{i = 1}^{n} x_i^{\alpha_i} = \bigoplus_{\alpha \in
\mathbb{F}_2^n} c_\alpha x^\alpha,
\end{equation}
This kind of expression of $f$ is called the Algebraic Normal
Form(ANF). The algebraic degree of $f$ is the number of variables in
the highest order term with nonzero coefficient, which is denoted by
$\deg(f)$.

A Boolean function is said to be symmetric if its output is
invariant under any permutation of its input bits. For a symmetric
Boolean function $f$ on $n$ variables, we have
\begin{equation}
f(x_1, x_2, \ldots, x_n) = f(x_{\sigma(1)}, x_{\sigma(2)}, \ldots,
x_{\sigma(n)})
\end{equation}
for all permutations $\sigma$ on $\{1, 2, \ldots, n\}$.

This equivalently means that the output of $f$ only depends on the
weight of its input vector. As a consequence, $f$ is related to a
function $v_f : \{0, 1, \ldots, n\} \mapsto \mathbb{F}_2$ such that
$f(\alpha) = v_f(\wt(\alpha))$ for all $\alpha \in \mathbb{F}_2^n$.
The vector $v_f = ($ $v_f(0)$, $v_f(1)$, $\ldots$, $v_f(n))$ is
called the simplified value vector(SVV) of $f$. The set of all
$n$-variable Boolean functions are denoted by $\mathcal{SB}_n$.

\begin{proposition}\cite{Canteaut_sbf}
A Boolean function $f$ on $n$ variables is symmetric if and only if
its ANF can be written as follows:
\begin{equation}
f(x_1, x_2, \ldots, x_n) = \bigoplus_{i=0}^{n}{\lambda_f(i)
\bigoplus_{\substack{\alpha \in \mathbb{F}_2^n\\\wt(\alpha) =
i}}{x^\alpha}} = \bigoplus_{i=0}^{n}{\lambda_f(i) \sigma_i^n},
\end{equation}
where $\sigma_i^n$ is the elementary symmetric polynomial of degree
$i$ on $n$ variables.
\end{proposition}

Then, the coefficients of the ANF of $f$ can be represented by a
$(n+1)$-bit vector, $\lambda_f = (\lambda_f(0), \lambda_f(1),
\ldots, \lambda_f(n))$, called the simplified algebraic normal
form(SANF) vector of $f$.

\begin{proposition}\cite{Canteaut_sbf} Let $f$ be a symmetric Boolean function on $n$
variables. Then, its simplified value vector $v_f$ and its
simplified ANF vector $\lambda_f$ are related by
\begin{equation}
v_f(i) = \bigoplus_{k \preceq i}{\lambda_f(k)} \text{ and }
\lambda_f(i) = \bigoplus_{k \preceq i}{v_f(k)},
\end{equation}
for all $i = 0, 1, \ldots, n$.
\end{proposition}

\begin{definition}\cite{Meier_aaadobf} For a given $f \in \mathcal{B}_n$, a nonzero function
$g \in \mathcal{B}_n$  is called an annihilator of $f$ if $f g=0$
and the algebraic immunity(AI) of $f$, is the minimum degree of all
annihilators of $f$ or $f\oplus 1$, which is denoted by
$\textsf{AI}(f)$.
\end{definition}

Note that $\AI(f) \leq \deg(f)$, since $f(f\oplus 1) = 0$.
Therefore, a function with high AI will not have a low algebraic
degree. It was known from \cite{Courtois1} that for any $f \in
\mathcal{B}_n$, $ \textsf{AI}(f) \leq \lceil \frac{n}{2} \rceil$.

Two Boolean functions $f$ and $g$ are said to be affine equivalent
if there exist $A \in GL_n(\mathbb{F}_2)$ and $b \in \mathbb{F}_2^n$
such that $g(x) = f(xA+b)$. Clearly, algebraic degree, algebraic
immunity are affine invariant.

The binary representation of an integer $a$ is denoted by $(a_m
a_{m-1} \ldots a_0)_2$, such that
\begin{equation}
a = \sum_{i = 0}^m{a_i 2^i}.
\end{equation}
If integer $b$ is ended by $a_1 a_0$ in binary, we often denote by
$b = (*a_1 a_0)_2$, where $*$ represents some $01$ string. For
convenience of the description in the sequel, we introduce the
following notation.

\begin{definition}
Let $a$, $b$ be two nonnegative integers with their binary
representations $(a_m a_{m-1} \ldots a_0)_2$ and $(b_n b_{n-1}
\ldots b_0)_2$, $m \leq n$. If $a_i = b_i$ for all $i = 0, 1, \ldots
m$, we say $a$ is a suffix of $b$ in binary and denote by $a
\preceq' b$. Furthermore, if $a < b$, we say $a$ is a proper suffix
of $b$, which is denoted by $a \prec' b$.
\end{definition}

\section{Main Results}

\begin{lemma}
\label{lemma_inverse} Let $f, g \in \mathcal{B}_n$, integer $0 \leq
d \leq n$. If $f(\alpha) = \bigoplus_{\substack{\beta \preceq
\alpha\\0 \leq \wt(\beta) \leq d}}{g(\beta)} $ for all $\alpha \in
\mathbb{F}_2^n$ with $\wt(\alpha)\leq d$, then $ \label{equ_inv2}
g(\beta) = \bigoplus_{\alpha \preceq \beta}{f(\alpha)} $ for
all$\beta \in \mathbb{F}_2^n$ with $\wt(\beta) \leq d$.
\end{lemma}
\begin{proof}
By direct computation, for any $\beta \in \mathbb{F}_2^n$ with
$\wt(\beta) \leq d$, we have
\begin{eqnarray*}
\bigoplus_{\alpha \preceq \beta}{f(\alpha)} & = & \bigoplus_{\alpha
\preceq \beta}{\bigoplus_{\gamma \preceq
\alpha}{g(\gamma)}}\\
& = & \bigoplus_{\gamma \preceq \beta}{\Big{(} g(\gamma)
\bigoplus_{\gamma \preceq \alpha \preceq \beta}{1} \Big{)} } \\
& = & \bigoplus_{\gamma \preceq \beta}{2^{\wt(\beta) -
\wt(\gamma)}g(\gamma)} = g(\beta),
\end{eqnarray*}
which completes our proof.
\end{proof}

\begin{lemma}
\label{lemma_majority} Let $f, g \in \mathcal{B}_n$, integer $0 \leq
d \leq n$. If $f(\alpha) = 1$ for all $\alpha \in \mathbb{F}_2^n$
satisfying $0 \leq \wt(\alpha) \leq d$ and $g(\beta) = 1$ for all
$\beta \in \mathbb{F}_2^n$ satisfying $n-d \leq \wt(\beta) \leq n$,
then both $f$ and $g$ do not have annihilators with degree less than
or equal to $d$.
\end{lemma}
\begin{proof}
Let $g' = g(x_1 \oplus 1, x_2 \oplus, \ldots, x_n \oplus 1)$, which
takes $1$ on all points with weight not exceeding $d$. Since $g'$ is
affine equivalent to $g$, $\AI(g') = \AI(g)$. Therefore, it suffices
to prove $f$ has no annihilator with degree not greater than $d$.

Assuming there is a function $h \in \mathcal{B}_n$ such that $f h =
0$ and $\deg(h) < d$, we will show that $h = 0$. Write $h$ in ANF
\begin{equation*}
h = \bigoplus_{\alpha \in \mathbb{F}_2^n}{c_{\alpha}x^\alpha}.
\end{equation*}
Since for any $\alpha \in \mathbb{F}_2^n$ with $\wt(\alpha) \leq d$,
we have $h(\alpha) = 0$, i.e., $\bigoplus_{\beta \preceq
\alpha}{c_\beta} = 0$. By Lemma \ref{lemma_inverse}, for any $\beta
\in \mathbb{F}_2^n$ with $\wt(\beta) \leq d$, $c_\beta =
\bigoplus_{\alpha \preceq \beta}{h(\alpha)} = 0$. Combining with
$\deg(h) \leq d$, we conclude $h = 0$.
\end{proof}

The following theorem is our main result, which gives a sufficient
condition for a function $f \in \mathcal{SB}_{2k}$ to have algebraic
immunity not less than $d$, where $d$ is a suffix of $k$ in binary.

\begin{theorem}
\label{thm_main} Let $f \in \mathcal{SB}_n$, $n = 2k$, $d \preceq'
k$ and $d \geq 2$. If for any integer $i, j$ with $0 \leq i \leq d -
1$, $n - d + 1 \leq j \leq n$ and
\begin{equation}
k - i \equiv j - k \equiv 2^t \, \mod 2^{t+1}
\end{equation} for some nonnegative integer $t$, $ v_f(i)
= v_f(j) \oplus 1$ holds, then $\AI(f) \geq d$.
\end{theorem}
\begin{proof}
To prove $\AI(f) \geq d$, we need to show $f$ or $f\oplus1$ has no
annihilator with degree less than $d$. Without loss of generality,
we only need to prove $f$ has no annihilator with degree less than
$d$, because it also satisfies the conditions in this theorem by
replacing $f$ by $f \oplus 1$.

Assume there is a function $g \in \mathcal{B}_n$, such that $f g =
0$ and $\deg(g) \leq d - 1$, our aim is to show $g = 0$. Write $g$
in ANF
\begin{equation*}
g = \bigoplus_{\alpha \in \mathbb{F}_2^n}{c_{\alpha}x^\alpha}.
\end{equation*}
Since $\deg(g) \leq d-1$, we have $c_\alpha = 0$ for all
$\wt(\alpha) \geq d$. If $f(\alpha) = 1$, then $g(\alpha) = 0$,
which is
\begin{equation}
\label{equ_sum0} \bigoplus_{\substack{\beta \preceq \alpha \\ 0 \leq
\wt(\beta) \leq d-1}}{c_\beta} = 0.
\end{equation}
Denote equation (\ref{equ_sum0}) on point $\alpha$ by $s_\alpha =
0$. By Lemma \ref{lemma_inverse}, we know $c_\beta = \oplus_{\alpha
\preceq \beta}{s_\alpha}$ for $\wt(\beta) \leq d-1$. We need to
prove that all the equations $s_\alpha=0$, $\alpha \in 1_f$, on
$\sum_{i=0}^{d-1}{\binom{n}{i}}$ variables $c_\beta$, $\wt(\beta)
\leq d-1$, has only zero solution.

To assist our proof, we introduce a decomposition of integers
according to $k$. Let $k=(k_m k_{m-1} \ldots k_0)_2$, then
\begin{equation}
C_p = \left\{
\begin{array}{ll}
\{ x \mid x - k \equiv 2^p \mod 2^{p+1} \}, & 0 \leq p \leq m, \\
\{ x \mid x - k \equiv 0 \mod 2^{m+1} \}, & p = m+1.
\end{array}
\right.
\end{equation}
In other words, $C_p$, $0 \leq p \leq m$ contains all integers with
binary representation $(*\overline{k_p}k_{p-1}\cdots k_0)_2$ and
$C_{m+1}$ contains all integers with binary representation $(*k_m
k_{m-1}\cdots k_0)_2$. It's easy to see $C_p$, $p = 0, 1, \ldots,
m+1$ is a decomposition of all integers and $[0, d-1] \cup [n-d+1,
n]\subseteq \cup_{i=0}^{\lfloor \log_2{d}\rfloor}{C_i}$.

For convenience of the following description, we define some
collections of equations, say $A_i, B_i$ and $E_i$, where
\begin{equation}
\left.
\begin{array}{l}
A_i  =  \{s_\alpha=0 \mid \alpha \in \mathbb{F}_2^n, \wt(\alpha) \in
[0, d-1] \text{ and }
\wt(\alpha) \in C_i \},\\
B_i  =  \{s_\alpha=0 \mid \alpha \in \mathbb{F}_2^n, \wt(\alpha)
\in [n-d+1, n] \text{ and } \wt(\alpha) \in C_i \},\\
E_i \in \{A_i, B_i\},
\end{array}
\right.
\end{equation}
for $i = 0, 1, \ldots, \lfloor \log_2{d} \rfloor$. Now, we use math
induction to prove that $A_0$ or $B_0$, union $A_1$ or $B_1$,
$\ldots$, union $A_p$ or $B_p$, denoted by $\cup_{i=0}^{p}E_i$, has
the same solution space with $\cup_{i=0}^{p} A_i$, i.e.,
$\spa(\cup_{i=0}^p{E_i}) = \spa(\cup_{i=0}^p{A_i})$, for $p = 0, 1,
\ldots, \lfloor \log_2{d} \rfloor.$ The induction parameter is $p$.

\textbf{Basis step: $p = 0$.} First, we will prove that the solution
space of $A_0$ is a subspace of that of $B_0$ by representing all
the equations in $B_0$ as linear combinations of equations in $A_0$.
Take an arbitrary equation $s_\alpha = 0$ in $B_0$, expanding
$s_\alpha$ as follows,
\begin{eqnarray}
s_\alpha & = & \bigoplus_{\substack{\beta \preceq \alpha \\
0 \leq \wt(\beta) \leq d-1}}{c_\beta}  =  \bigoplus_{\substack{\beta \preceq \alpha \\
0 \leq \wt(\beta) \leq d-1}}{\bigoplus_{\gamma \preceq
\beta}{s_\gamma}} \nonumber\\
&= & \bigoplus_{\substack{\gamma \preceq \alpha\\0 \leq \wt(\gamma)
\leq d - 1}}{\Big{(}s_\gamma \bigoplus_{\substack{\gamma \preceq
\beta \preceq \alpha\\0 \leq \wt(\beta) \leq d-1}}{1} \Big{)}} \nonumber\\
\label{equ_lastc} & = & \bigoplus_{\substack{\gamma \preceq
\alpha\\0 \leq \wt(\gamma) \leq d-1}}{\Bigg{(} s_\gamma \bigoplus_{i
= 0}^{d - 1 - \wt(\gamma)}{\binom{\wt(\alpha) - \wt(\gamma)}{i}}
\Bigg{)}}.
\end{eqnarray}

Considering $s_\gamma$ in the \eqref{equ_lastc}, where $\wt(\gamma)
\not\in C_0$, we want to show the coefficient of $s_\gamma$ is $0$.
By Lucas' formula, we know $\binom{\wt(\alpha)-\wt(\gamma)}{i}=1$
over $\mathbb{F}_2$ if and only if $i \preceq
\wt(\alpha)-\wt(\gamma)$. Note that $\wt(\alpha) - \wt(\gamma) = (*
\overline{k_0})_2 - (* k_0)_2 = (* 1)_2$ and $d - 1 - \wt(\gamma) =
(* k_0)_2 - 1 - (* k_0)_2 = (*1)_2$. Hence, if $i = (\cdots i_2 i_1
0)_2$ satisfies $i \preceq \wt(\alpha) - \wt(\gamma)$ and $i \leq d
- 1 - \wt(\gamma)$, then $i+1 = (\cdots i_2 i_1 1)_2$ also satisfies
the above constraints and vice versa. We conclude that an $i$ ended
by $0$ in its binary representation satisfying $i \preceq
\wt(\alpha)-\wt(\gamma)$ must correspond with another $i$ ended by
$1$ in the inner sum of \eqref{equ_lastc}. Thus, $\bigoplus_{i =
0}^{d - 1 - \wt(\gamma)}{\binom{\wt(\alpha) - \wt(\gamma)}{i}} = 0$
when $\gamma \not \in C_0$, and all equations in $B_0$ could be
represented as linear combinations of those in $A_0$. Therefore a
solution of equations $A_0$ is also a solution of $B_0$, which
implies the solution space of $A_0$ is a subspace of that of $B_0$.

By Lemma \ref{lemma_majority}, it's easy to see equations in both
$A_0$ and $B_0$ are linearly independent. Since they have the same
size, the dimensions of both solution spaces are the same.
Therefore, the solution spaces of $A_0$ and $B_0$ are the same,
which completes the basis step for $p = 0$.

\textbf{Induction step: assuming the proposition is true for $p = q
- 1$, $q \geq 1$, we will prove it's also true for $p = q$.}

First, we will prove the solution space of $\cup_{i = 0}^{q}A_i$ is
a subspace of that of $\cup_{i = 0}^{q-1}{A_i} \cup B_q$. Taking an
arbitrary $s_\alpha = 0$ in $B_q$, we want to show $s_\alpha$ can be
represented as linear combinations of equations in $\cup_{i =
0}^{q}A_i$. Similar with the method in basis step, expand $s_\alpha$
as
\begin{equation}
\label{equ_lastc2} \bigoplus_{\substack{\gamma \preceq \alpha\\0
\leq \wt(\gamma) \leq d-1}}{\Bigg{(} s_\gamma \bigoplus_{i = 0}^{d -
1 - \wt(\gamma)}{\binom{\wt(\alpha) - \wt(\gamma)}{i}} \Bigg{)}}.
\end{equation}
The key is to show $\bigoplus_{i = 0}^{d - 1 -
\wt(\gamma)}{\binom{\wt(\alpha) - \wt(\gamma)}{i}} = 0$ when
$\wt(\gamma) \not \in \cup_{i = 0}^{q}{C_i}$. Take an arbitrary
$\gamma$ such that $\wt(\gamma) \not \in \cup_{i = 0}^{q}{C_i}$.
Noting that $\wt(\alpha) = (*\overline{k_q} k_{q-1} \cdots k_0)_2$,
$\wt(\gamma) = (*k_q k_{q-1} \ldots k_0)_2$ and $d = (k_{\lfloor
\log_2{d}\rfloor}\cdots k_q k_{q-1} \cdots k_0)_2 - 1$, we have
$\wt(\alpha) - \wt(\gamma) = (*1\underbrace{0\cdots0}_{q \text{
times}})_2$ and $d - 1 - \wt(\gamma) = (*1\underbrace{1\cdots1}_{q
\text{ times}})_2$. It's easy to see that if there is an $i =
(*0i_{q-1} \cdots i_0)_2$, $0 \leq i \leq d - 1 - \wt(\gamma)$,
satisfying $\binom{\wt(\alpha) - \wt(\gamma)}{i} = 1$, i.e., $i
\preceq \wt(\alpha) - \wt(\gamma)$, then $i + 2^q = (*1i_{q -
1}\cdots i_0)_2$ also satisfies $i + 2^q \preceq \wt(\alpha) -
\wt(\gamma)$ and $i + 2^q \leq d - 1 - \wt(\gamma)$ and vice versa.
Since this correspondence is one on one, the $1's$ in the inner sum
of \eqref{equ_lastc2} can be divided into pairs. Therefore,
$\bigoplus_{i = 0}^{d - 1 - \wt(\gamma)}{\binom{\wt(\alpha) -
\wt(\gamma)}{i}} = 0$ and all equations in $B_q$ can be written as
sums of equations in $\cup_{i=0}^q{A_i}$. We conclude that the
solution space of $\cup_{i = 0}^{q}A_i$ is a subspace of that of
$\cup_{i = 0}^{q-1}{A_i} \cup B_q$.

By induction hypothesis,
\begin{equation*}
\spa(\cup_{i = 0}^{q-1}{A_i} \cup B_q) = \spa(\cup_{i =
0}^{q-1}{B_i} \cup B_q) = \spa(\cup_{i = 0}^{q}{B_i}).
\end{equation*} And by Lemma \ref{lemma_majority}, it's not hard to see there is no linear
dependence in $\cup_{i = 0}^{q}{B_i}$ as well as in $\cup_{i =
0}^{q}{A_i}$. Note that $|\cup_{i=0}^{q}{A_i}| =
|\cup_{i=0}^{q}{B_i}|$, the dimensions of the solution spaces of
$\cup_{i = 0}^{q}A_i$ and $\cup_{i = 0}^{q-1}{A_i} \cup B_q$ are the
same. Combining with the fact that solution space of $\cup_{i =
0}^{q}A_i$ is a subspace of that of $\cup_{i = 0}^{q-1}{A_i} \cup
B_q$, we claim these two solution spaces are exactly the same. Using
induction hypothesis again, we have
\begin{eqnarray*}
\spa(\cup_{i=0}^{q}{A_i}) & = & \spa(\cup_{i = 0}^{q-1}{A_i} \cup
B_q) \\ & = & \spa(\cup_{i = 0}^{q-1}{E_i} \cup B_q) \\ & = &
\spa(\cup_{i = 0}^{q}{E_i}),
\end{eqnarray*}
which completes the induction.

Now, let's go back to the original problem that proving $g = 0$. By
the conditions in this theorem, for any $\alpha \in \mathbb{F}_2^n$,
$\wt(\alpha) \in C_t \cap [0, d - 1]$, we have $f(\alpha) = m$; for
any $\alpha \in \mathbb{F}_2^n$, $\wt(\alpha) \in C_t \cap [n-d+1,
n]$, we have $f(\alpha) = m \oplus 1$, where $m = 0$ or $1$. If $m =
1$, we could list equations on the point $\alpha$, where
$\wt(\alpha) \in C_t \cap [0, d-1]$, which is exactly the equations
set $A_t$. If $m = 0$, we could list equations on the point
$\alpha$, where $\wt(\alpha) \in C_t \cap [n-d+1, n]$, which is
exactly the equations set $B_t$. If let $t$ run over from $0$ to
$\lfloor \log_2{d} \rfloor$, we obtain equations
$\cup_{i=0}^{\lfloor \log_2{d} \rfloor}E_i$, which is equivalent to
$\cup_{i=0}^{\lfloor \log_2{d} \rfloor}A_i$. By Lemma
\ref{lemma_majority}, $\cup_{i=0}^{\lfloor \log_2{d} \rfloor}A_i$
has only zero solution, thus $\cup_{i=0}^{\lfloor \log_2{d}
\rfloor}E_i$ has only zero solution. Therefore, $g=0$ and the proof
is complete.
\end{proof}

\begin{construction}
\label{cons} Given two positive integers $k$, $d$, where $d \preceq'
k$ and $2 \leq d \leq k$, we construct a function $f$ in
$\mathcal{SB}_{2k}$ as follows.
\begin{itemize}
\item Choose $\lfloor \log_2{d} \rfloor + 1$ numbers in
$\mathbb{F}_2$ arbitrarily, denoted by $m_0$, $m_1$, $\ldots$,
$m_{\lfloor \log_2{d} \rfloor}$.
\item Define a symmetric Boolean function $f$ through it's simplified value vector, which is
\begin{equation}
v_f(i) = \left \{
\begin{array}{ll}
m_t, & i \in C_t \cap [0, d-1], \\
m_t \oplus 1, & i \in C_t \cap [n-d+1, n], \\
0 \text{ or } 1, & \text{otherwise}.
\end{array}
\right.
\end{equation}
\end{itemize}
\end{construction}

By Theorem \ref{thm_main}, $\AI(f) \geq d$ for $f$ in Construction
\ref{cons}. We present an example here to illustrate our
construction. Let $k = 6 = (110)_2$ and $d = k$. We have $C_0 = \{1,
3, 5, 7, 9, 11, \ldots\}$, $C_1 = \{0, 4, 8, 12, \ldots\}$ and $C_2
= \{2, 10, \ldots\}$. Therefore, constraints $v_f(1) = v_f(3) =
v_f(5) = v_f(7) \oplus 1 = v_f(9) \oplus 1= v_f(11) \oplus 1$,
$v_f(0) = v_f(4) = v_f(8) \oplus 1 = v_f(12) \oplus 1$ and $v_f(2) =
v_f(10) \oplus 1$ must be satisfied. Let $m_0, m_1, m_2 \in
\mathbb{F}_2$ take over all the $8$ combinations, we obtain the
following $8$ functions with maximum algebraic immunity in Table 1.

\begin{table}[h]
\caption{Functions in $\mathcal{SB}_8$ with maximum AI} \centering
\begin{tabular}{c|c|c}
\hline
$m_0 m_1 m_2$ & SVV:$v_f(0)\ldots v_f(12)$ & SANF:$\lambda_f(0) \ldots \lambda_f(12)$\\
\hline
       000 & 0000000111111 & 0000000110000\\
\hline
       000 & 0000001111111 & 0000001010000 \\
\hline
       001 & 0010000111011 & 0011001010000 \\
\hline
       001 & 0010001111011 & 0011000110000 \\
\hline
       010 & 1000100101110 & 1111000110000 \\
\hline
       010 & 1000101101110 & 1111001010000 \\
\hline
       011 & 1010100101010 & 1100001010000 \\
\hline
       011 & 1010101101010 & 1100000110000 \\
\hline
       100 & 0101010010101 & 0100000110000 \\
\hline
       100 & 0101011010101 & 0100001010000 \\
\hline
       101 & 0111010010001 & 0111001010000 \\
\hline
       101 & 0111011010001 & 0111000110000 \\
\hline
       110 & 1101110000100 & 1011000110000\\
\hline
       110 & 1101111000100 & 1011001010000 \\
\hline
       111 & 1111110000000 & 1000001010000 \\
\hline
       111 & 1111111000000 & 1000000110000 \\

\end{tabular}
\end{table}

\begin{corollary}
The number of symmetric Boolean functions on $2k$ variables, with
algebraic immunity greater than or equal to $d$, $d \geq 2$ and $d
\preceq' k$, is not less than
\begin{equation}
2^{\lfloor \log_2{d} \rfloor + 2(k-d+1)}.
\end{equation}
\end{corollary}
\begin{proof}
We prove this by enumerating all the functions in Construction
\ref{cons}. There are $\lfloor \log_2{d} \rfloor + 1$ numbers on
$\mathbb{F}_2$ could be chosen arbitrarily. To show different
choices will generate different functions, it's sufficient to prove
$C_t \cap [0, d-1] \not = \emptyset$. If $0 \leq t \leq \lfloor
\log_2{d} \rfloor - 1$, it's obvious that $(\overline{k_t} \cdots
k_1 k_0)_2 \in C_t$ and $(\overline{k_t} \cdots k_1 k_0)_2 <
(k_{\lfloor \log_2{d} \rfloor}\cdots k_1 k_0)_2 = d$. If $t =
\lfloor \log_2{d} \rfloor$, $(\overline{k_t} \cdots k_1 k_0)_2 \in
C_t$. Because $k_t = 1$, we have  $(\overline{k_t} \cdots k_1
k_0)_2< d$.

Since the number of all choices for $m_0, m_1, \ldots, m_{\lfloor
\log_2{d} \rfloor}$ is $2^{\lfloor \log_2{d} \rfloor + 1}$ and
$v_f(i)$ could take either $0$ or $1$ when $i \in [d, n-d]$, the
total number of such of $f$ can be constructed is
\begin{equation*}
2^{\lfloor \log_2{d} \rfloor + 1 + n-d-d + 1} = 2^{\lfloor \log_2{d}
\rfloor + 2(k-d+1)},
\end{equation*}
which completes our proof.
\end{proof}

We present another example here to illustrate our counting result.
Let $k = 13 = (1101)_2$, $d=5=(101)_2 \prec' k$. Hence $C_0 = \{0,
2, 4, 6, \ldots \}$, $C_1 = \{3, 7, \ldots \}$ and $C_2 = \{1, 9,
\ldots \}$. For arbitrary $m_0, m_1, m_2 \in \mathbb{F}_2$, $m_0 =
v_f(0) = v_f(2) = v_f(4) = v_f(26)\oplus 1 = v_f(24) \oplus 1 =
v_f(22) \oplus 1$, $m_1 = v_f(3) = v_f(23)\oplus1$ and $m_2 = v_f(1)
= v_f(25)\oplus1$ must be satisfied, while the others bits could
take $0$ or $1$ arbitrarily. Let $m_0, m_1, m_2$ run over all $8$
combinations, $2^{20}$ functions $\in \mathcal{SB}_{26}$ are
constructed and listed in Table 2.

\begin{table}[h]
\caption{Functions in $\mathcal{SB}_{26}$ with AI not less than $5$}
\centering
\begin{tabular}{c|c}
\hline
$m_0 m_1 m_2$ & SVV:$v_f(0) v_f(1)\ldots v_f(26)$\\
\hline
       000 & 00000???$\cdots$???11111\\
\hline
       001 & 01000???$\cdots$???11101\\
\hline
       010 & 00010???$\cdots$???10111\\
\hline
       011 & 01010???$\cdots$???10101\\
\hline
       100 & 10101???$\cdots$???01010\\
\hline
       101 & 11101???$\cdots$???01000\\
\hline
       110 & 10111???$\cdots$???00010\\
\hline
       111 & 11111???$\cdots$???00000\\
\end{tabular}
\end{table}


\begin{thebibliography}{99}
\bibitem{Carlet1}
C.~Carlet, D.~K. Dalai, K.~C. Gupta, and S.~Maitra,
\emph{``Algebraic immunity for cryptographically significant Boolean
functions: analysis and construction''}, IEEE Trans. on Information
Theory, vol.52, no.7, pp.3105-3121, JULY 2006.

\bibitem{improvingfastalgebraicattacks_Armknecht}
F.~Armknecht.,\emph{``Improving fast algebraic attacks''}, In FSE
2004, vol.3017 of Lecture Notes in Computer Science, pp.65-82,
Spring-Verlag, 2004.

\bibitem{openproblems}
A.~Canteaut, \emph{``Open problems related to algebraic attacks on
stream ciphers''}, in Proc. WCC 2005, Invited talk, pp.1-10.

\bibitem{basictheory_Dalai} D.~K. Dalai, S.~Maitra, and S.~Sarkar,
\emph{``Basic theory in construction of Boolean functions with
maximum possible annihilator immunity"}, Des. Codes, Cryptogr., vol.
40, no.1, pp.41-58, 2006.

\bibitem{LiNa1}
N.~Li, L.~Qu, W.~Qi, G.~Feng, C.~Li and D.~Xie, \emph{``On the
construction of Boolean functions with optimal algebraic
immunity''}, IEEE Trans. on Information Theory, vol.54, no.3,
pp.1330-1334, MARCH 2008.

\bibitem{Courtois1}
N.~Courtois and W.Meier, \emph{``Algebraic attacks on stream ciphers
with linear feedback''}, in Advances in Cryptology - EUROCRYPT 2003.

\bibitem{Meier_aaadobf}
W.~Meier, E.~Pasalic, and C.~Carlet, \emph{``Algebraic attacks and
decomposition of Boolean functions,''} in Advances in
Cryptology-EUROCRYPT 2004. Berlin, Germany: Springer-Verlag, 2004,
vol.3027, Lecture Notes in Computer Science, pp.474-491.

\bibitem{Qu2}
L.~Qu, K. Feng, F. Liu, and L. Wang, \emph{``Constructing symmetric
Boolean function with maximum algebraic immunity''}, IEEE Trans. on
Information Theory, vol.55, no.5, pp.2406-2412, MAY, 2009.

\bibitem{Canteaut_sbf}
A. Canteaut and M. Videau, \emph{``Symmetric Boolean functions''},
IEEE Trans. on Information Theory, vol.51, no.8, pp.2791-2811, Aug.,
2005.

\bibitem{construction_Carlet}
C.~Carlet and P.~Gaborit, \emph{``On the construction of balanced
Boolean functions with a good algebraic immunity,''} in Proc.2005
Int. Wrokshop on Boolean Functions: Cryptogr. Appl. , Rouen, France,
Mar., 2005, pp.1-14.

\bibitem{anewupperbound}
F.~Didier, \emph{``A New Upper Bound on the Block Error Probability
After Decoding Over Erasure Channel''}, IEEE Trans. on Information
Theory, vol.52, no.10, pp.4496-4503, Oct., 2006.

\bibitem{LiNa2}
N.~Li, W.~Qi and K.~Feng, \emph{``Symmetric Boolean functions
depending on an odd number of variables with maximum algebraic
immunity''}, IEEE Trans. on Information Theory, vol.52, no.5,
pp.2271-2273, MAY 2006.

\bibitem{Qu1}
L.~Qu, C. Li and K. Feng, \emph{``A note on symmetric Boolean
functions with maximum algebraic immunity in odd number
ofvariables''}, IEEE Trans. on Information Theory, vol.53, no.8, pp.
2908-2910, Aug. 2007.

\bibitem{Liao_c2mp1}
Q.~Liao, F.~Liu and K.~Feng, \emph{``On $(2^m+1)$-variable symmetric
Boolean functions with submaximum algebraic immunity $2^{m-1}$''},
Science in China Series A: Mathematics, vol.52, no.1, pp.17-28,
Jan., 2009.

\bibitem{Braeken}
A.~ Braeken, \emph{``Cryptographic Properties of Boolean functions
and S-Boxes''}, thesis, Mar., 2006.

\end{thebibliography}
\end{document}